\documentclass[runningheads]{llncs}

\usepackage[mathscr]{eucal}
\usepackage{float}
\usepackage{pstricks}
\usepackage{pst-coil}
\usepackage[all]{xy}
\newlength{\pu}
\spnewtheorem{notation}{Notation}{\bfseries}{\itshape}
\newcommand{\midoc}{dynamic I/O automaton }  
\newcommand{\debut}{\mbox{\textrm{begin}}}
\newcommand{\fin}{\mbox{\textrm{end}}}
\newcommand{\cC}{\gamma}
\floatstyle{ruled}
\newfloat{Module}{thp}{lop}
\newcommand{\remove}[1]{}
\begin{document}

\mainmatter \title{A Framework for Proving the Self-organization of Dynamic Systems}

\author{Emmanuelle Anceaume\inst{1} \and Xavier D\'efago\inst{2} \and Maria Potop-Butucaru\inst{3} \and Matthieu Roy\inst{4} }
\authorrunning{E. Anceaume \and X. D\'efago\inst{2} \and M. Potop-Butucaru\inst{3} \and M. Roy\inst{4} }

\institute{CNRS UMR 6074, IRISA, Rennes, France\\
  \email{Emmanuelle.Anceaume@irisa.fr}
  \and School of Information Science, JAIST, Nomi, Ishikawa, Japan\\
  \email{defago@jaist.ac.jp}
  \and LIP6, INRIA-Universit\'e Paris 6, France\\
  \email{maria.gradinariu@lip6.fr}
  \and LAAS-CNRS, Toulouse, France\\
  \email{matthieu.roy@laas.fr} }

\maketitle

\begin{abstract}
  This paper aims at providing a rigorous definition of 
  \emph{self-organization}, one of the most desired properties for
  dynamic systems (e.g.,  peer-to-peer systems, sensor networks,
  cooperative robotics, or ad-hoc networks).  We characterize different classes of self-organization through  liveness and safety properties that both  capture information regarding the system entropy. We illustrate these classes through study cases. The first ones are  two
  representative P2P overlays (CAN and Pastry) and the others are specific
  implementations of $\Omega$ (the leader oracle)
  and one-shot query abstractions for dynamic settings.  Our study aims at understanding
  the limits and respective power of existing self-organized protocols
  and lays the basis of designing robust algorithm for dynamic
  systems.
\end{abstract}

\section{Introduction}
\label{sec:intro}

Self-organization is an evolutionary process that appears in many
disciplines. Physics, biology, chemistry, mathematics, economics, just
to cite a few of them, show many examples of self-organizing
systems. Crystallization, percolation, chemical reactions, proteins
folding, flocking, cellular automata, market economy are among the
well-admitted self-organizing systems. In all these disciplines,
self-organization is described as a process from which properties
emerge at a global level of the system. These properties are solely
due to local interactions among components of the system, that is with
no explicit control from outside the system. Influence of the
environment is present but not intrusive, in the sense that it does
not disturb the internal organization process.

In the newly emerging fields of distributed systems (peer-to-peer systems, ad-hoc
networks, sensors networks, cooperative robotics), self-organization
becomes one of the most desired properties.  The major feature of all
recent scalable distributed systems is their extreme dynamism in terms
of structure, content, and load.  In peer-to-peer systems (often referred as to P2P systems), nodes continuously
join and leave the system. In large scale sensor, ad-hoc or robot
networks, the energy fluctuation of batteries and the inherent
mobility of nodes induce a dynamic aspect of the
system. 
In all these systems there is no central entity in charge of their
organization and control, and there is an equal capability, and
responsibility entrusted to each of them to own
data~\cite{harnessing}. To cope with such characteristics, these
systems must be able to spontaneously organize toward desirable global
properties.  In peer-to-peer systems, self-organization is handled
through protocols for node arrival and departure, as provided by
distributed hash tables based-overlay (e.g.,~CAN~\cite{CAN-thesis}, Chord~\cite{chord},
Pastry~\cite{DR01,RD01}), or random graph-based ones
(e.g.,~Gnutella~\cite{gnutella}, GIA~\cite{gia}). 
Secure operations in ad-hoc networks rely on a self-organizing
public-key management system that allows users to create, store,
distribute, and revoke their public keys without the help of any
trusted authority or fixed server~\cite{CBH03}. Recent large scale
applications (e.g.,~persistent data store, monitoring) exploit the
natural self-organization of peers in semantic communities to provide
efficient search~\cite{GBF04,MJB04,SMZ03}.  Self-organizing heuristics
and algorithms are implemented in cooperating  robots networks  so that
they can configure themselves into predefined geometric patterns
(e.g.,~\cite{FMSY02,SY99,ZA02}). For instance, crystal growth
phenomenon inspired Fujibayashi et~al.~\cite{FMSY02} to make robots
organize themselves into established shapes by mimicking spring
properties.

Informal definitions of self-organization, and of related self$^\ast$
properties (e.g., self-configuration, self-healing or
self-reconfiguration) have been proposed previously~\cite{BMM02,WWA00,ZA02}.  Specifically, Babao\u{g}lu et~al.~\cite{BMM02} propose a
platform, called Anthill, whose aim is the design of  peer-to-peer
applications based on self-organized colonies and swarms of agents.
Anthill offers a bottom-up opportunity to understand the emergent
behavior of complex adaptive systems.  Walter et~al.~\cite{WWA00}
focus on the concept of reconfiguration of a metamorphic robotic
system with respect to a target configuration.
Zhang and Arora~\cite{ZA02} propose the concepts of self-healing and
self-configuration in wireless ad-hoc networks, and propose
self-stabilizing solutions for self$^\ast$-clustering in
ad-hoc networks~\cite{Dol00}. Note that a
comprehensive survey of self-organizing systems that emanate from
different disciplines is proposed by Serugundo et al.~\cite{self-org-survey}.


The main focus of this paper is to propose a formal
specification of the self-organization notion which, for the best of
our knowledge, has never been formalized in the area of scalable and
dynamic systems, in spite of an overwhelming use of the term.
Reducing the specification of self-organization to the behavior of the
system during its non-dynamic periods is clearly not satisfying
essentially because these periods may be very short, and rare. On the
other hand, defining self-organization as a simple convergence process
towards a stable predefined set of admissible configurations is
inadequate for the following two reasons. First, it may not be possible
to clearly characterize the set of admissible configurations since, in
dynamic systems, a configuration may include the state of some key
parameters that may not be quantified \emph{a priori} but have a strong influence on the dynamicity of the
system. For instance, the
status of batteries in sensor networks, or the data stored at peers in P2P
systems. Second, due to the dynamic behavior of nodes, it may happen
that no execution of the system converges to one of the predefined
admissible configurations.

Hence our attempt to specify self-organization according to the very
principles that govern dynamic systems, namely high interaction,
dynamics, and heterogeneity.  High interaction relates to the
propensity of nodes to continuously exchange information or
resources with other nodes around them.  For instance, in mobile
robots systems, infinitely often robots query their neighborhood to
arrange themselves within a predefined shape.  The dynamics of these
systems refer to the capability of nodes to continuously move around,
to join or to leave the system as often as they wish based on their
local knowledge. Finally heterogeneity refers to the specificity of
each entity of the system: some have huge computation resources, some
have large memory space, some are highly dynamic, some have broad
centers of interest. In contrast, seeing such systems as a simple mass
of entities completely obviates the differences that may exist between
individual entities; those very differences make the richness of these
systems.

The tenets mentioned above share as a common seed the locality
principle, i.e.,~the fact that both interactions and knowledge are
limited in range.  The first contribution of this paper is a
formalization of this idea, leading first to the notion of \emph{local
  self-organization}.  Intuitively, a locally self-organizing system
should reduce locally the entropy of the system.  For example, a
locally self-organized P2P system forces components to be adjacent to
components that improve, or at least maintain, some property or
evaluation criterion.  
The second contribution of the paper is the proposition of different classes of self-organization through safety and liveness properties that both capture information regarding the entropy of the system.  Basically, we propose three classes of self-organization. The first one characterizes  dynamic systems that converge toward sought global properties only during stable periods of time (these properties can be lost due to  instability). The second  one  depicts dynamic systems capable of  infinitely often increasing the convergence  towards global properties (despite some form of instability). Finally, the last one describes dynamic systems capable of continuously increasing that convergence. We show that complex emergent properties can be described as a combination of local and independent properties. 


The remaining of this paper is organized as follows: Section~\ref{sec:Model} proposes a model for dynamic and scalable systems, and in particular enriches the family of demons to capture the churn of the system.
Section~\ref{sec:discussions} proposes an overview of 
self-stabilization and superstabilization models.  
Section~\ref{sec:fromlocaltoself} formalizes the notion of self-organization
by introducing local and global evaluation criteria. Section~\ref{sec:different-level} characterizes different classes of self-organization through liveness and safety properties. Each class is illustrated with different study cases. Section~\ref{sec:composition} extends the concept of self-organization for one criteria to the one of self-organization for a finite set of independent criteria. 
Section \ref{sec:conclusion} concludes and discusses open
issues.

\section{Model}
\label{sec:Model}

\subsection{Dynamic System Model}
\subsubsection{Communication Graph.}
\label{sec:dynamic-system-model}
The physical network is described by a weakly connected graph. Its
nodes represent processes of the system and its edges represent
communication links between processes. In the following, this graph is referred  as to the physical communication graph.  We assume that the
physical communication graph is subject to frequent and unpredictable
changes. Causes of these changes are either due to volunteer actions
(such as joins and leaves of nodes in P2P systems, moving
of a robot in cooperating robots networks, or power supply limitation in
sensors networks) or accidental events (such as nodes failures in a
network of sensors, the sudden apparition of an obstacle that may temporarily
hide some region of a network of robots, or messages losses in a P2P
system). Therefore, even if the number of nodes in the system remains
constant, the physical graph communication may change due to
modifications in the local connectivity of nodes.

\subsubsection{Data Model.}  Nearly all modern applications in the
dynamic distributed systems are based on the principle of data
  independence---the separation of data from the programs that use
the data.  This concept was first developed in the context of database
management systems.
In dynamic systems, in particular in P2P systems, the presence of data stored locally
at each node 
 plays a crucial role in creating semantic based
communities. 
As a consequence, data are subject to frequent and unpredictable
changes to adjust to the dynamics of the nodes. In addition, for efficiency reasons it may be necessary to aggregate data, or even to erase part of them. Consequently data can be
subject to permanent or transient failures.

\subsubsection{Logical communication graph.}
\label{par:overlays}
The physical communication graph combined with  the data stored in the network
represent the (multi-layer) logical communication graph of the system,
each logical layer~$l$ being weakly connected.  In order to connect to
a particular layer~$l$, a process executes an underlying connection
protocol.  A node~$n$ is called \emph{active} at a layer~$l$ if there
exists at least one node~$r$ which is connected at $l$ and aware of
$n$.  The set of logical \emph{neighbors} of a node~$n$ at a layer~$l$
is the set of nodes~$r$ such that the logical link~$(n,r)$ is up ($n$
and $r$ are aware of each other) and is denoted ${\mathcal N}^l(n)$.
Notice that node~$n$ may belong to several layers simultaneously.
Thus, node $n$ may have different sets of neighbors at different logical
layers.  In peer-to-peer systems, the (multi-layer) communication
graph is often referred to as the logical structured overlay when the communication graph is structured (as it it the case with CAN~\cite{CAN-thesis}, Pastry~\cite{pastry}, and 
Chord~\cite{chord}), or unstructured logical overlay when the graph is random (as for example Gnutella, and Skype).  In sensors or ad-hoc networks,
connected coverings (such as trees, weakly connected maximal
independent sets or connected dominating sets) represent the logical
communication graph.

Note that in sensors networks, cooperating robots networks, and ad-hoc networks  the logical communication graph may coincide with the
physical communication one. On the other hand, P2P logical overlays do not usually share relationships with the physical location of peers.

\subsection{State Machine-Based Framework}
\label{sec:io}

To rigorously analyze the execution of the dynamic systems, we use the
dynamic I/O automata introduced by Attie and Lynch~\cite{AL01}. This
model allows the modeling of individual components, their interactions
and their changes.  The external actions of a dynamic I/O automata are
classified into three categories, namely the actions that modify data
(by replication, aggregation, removal, or writing), the input-output
actions (I/O actions), and the actions reflecting the system dynamics.
Regarding this last category, we identify connection (C) and
disconnection (D) actions. These actions model the arrival and departure of a node
 in a P2P system, or the  physical moving of a robot by decomposing its movement as 
disconnection followed by connection actions, or the removal 
of a too far away sensor from the view of a sensor (because of power supply limitation).   A {\it
  configuration} is the state of the system at time $t$ together with the
communication graph and data stored in the system.  An {\it execution}
is a maximal sequence of totally ordered configurations. Let $c_{t}$ be the current configuration of the system. It moves  to  configuration $c_{t+1}$ after the
execution of either an internal, an input/output action,  or a dynamic action.  A {\it fragment of execution} is a
finite sub-sequence of an execution.  Its size is the subsequence
length.  A \emph{static fragment} is a maximal-sized fragment that
does not contain any C/D actions. Let $f=(c_i, \ldots, c_j)$ be a fragment
in a \midoc execution.  We denote as $\mbox{\textrm{begin}}(f)$ and
$\mbox{\textrm{end}}(f)$ the configurations $c_i$ and $c_j$
respectively.  In the sequel, all the referred fragments are static
and are denoted by $f$ or $f_i$. Thus, an execution of a \midoc is a
infinite sequence of fragments $e=(f_0, \ldots,f_i, \ldots f_j,
\ldots)$ interleaved with dynamic actions.

\subsection{Churn Model}
\label{sec:dynamicschedulers}
As previously described, both the logical communication graph and the
system data are subject to 
changes. Whatever their causes, we refer to these changes as
to the churn of the system. Note that this definition is compliant with the one proposed by Brighten Godfrey et al~\cite{minimizingchurn} which defines churn as the sum, over each configuration, of the fraction of the system that has changed during that configuration, normalised by the duration of the execution.


Interestingly, the effect of the churn on the system computations is
similar with the effect of  synchrony in classical
distributed systems and in particular in the theory of
self-stabilization where the demon/scheduler
abstraction captures the system synchrony~\cite{Dol00}. According to that
synchrony, a hierarchy of schedulers that ranges from centralized to
arbitrary ones exists.  In short a demon is a predicate defined
over the forest of executions (in the sense defined above) of a
system. A system under a scheduler is the restriction of the
executions of the system to those verifying the scheduler predicate.
 
To capture the churn of the system we enrich the demons family with a
new class of demons --- the dynamic demons or dynamic schedulers. The
predicate of a dynamic demon characterises the dynamic actions (C/D) 
of the system. 

  


In the following we propose four classes  of dynamic demons. Note that
for the best of our knowledge, all churn models discussed so far in
distributed systems area are covered by these four classes~\cite{opodis05,baldoni07,minimizingchurn,hugues&carole,MRTPAA05}.

\begin{description}
\item[Bounded dynamic demon] Under this demon, the number of
C/D actions during any execution of  the system 
is finite and a bound on this number is known. In \cite{MRTPAA05},
the authors consider an $\alpha$-parameterized system where the
$\alpha$ parameter is a known  lower bound on the number of nodes that 
eventually remain connected within the system. In this model augmented with some 
communication synchrony the authors propose  the implementation of the $\Omega$ oracle discussed 
later in this paper. 
\item[Finite dynamic demon] Under this demon, the number of
C/D actions 
during any execution of the system is finite but the bound is unknown. That is, 
protocols designed under this demon cannot use the bound on the number of nodes active in the system. 
Different forms of the finite dynamic scheduler have been proposed in \cite{aguilera04}. 
Delporte et al \cite{hugues&carole} 
propose agreement algorithms compliant with this demon. 
\item[Kernel-based dynamic demon] Under this demon, between any two successive static fragments $f_{i}$ and $f_{i+1}$ of any execution of the system, there exists a non-empty group of nodes $\mathcal G_{i}$ for which the local knowledge is not impacted by C/D actions. 
Baldoni et al~\cite{baldoni07} extend this characterization to specific topological 
properties of the communication graph. 
Specifically, the authors define the churn impact  with respect to  the diameter of the graph,
and distinguish three classes of churn. In the first one,  the diameter of the graph is constant 
and every active node in the system can use this knowledge. In the second class, the diameter 
is upper bounded but active nodes do not have access to that upper bound. 
Finally, in the third one the diameter of the 
system can grow infinitely.   
\item[Arbitrary dynamic demon] Under this demon, there is no restriction on the number of C/D actions that can occur during any execution of the system. That is, at any time, any subset of nodes
can join or leave the system and  there is no stabilization
constraint on the communication graph.
\end{description}

In the following we discuss different candidate models for self-organization and further propose our model. In this framework we further revisit  the lookup primitives of CAN~\cite{CAN-thesis} and Pastry~\cite{pastry}, 
the leader election problem~\cite{MRTPAA05}, and the one-shot query problem~\cite{baldoni07}. We  show the level of self-organization solutions to these problems ensure according to the strength of the dynamic demon.

\section{Candidate Models for Self-organization}
\label{sec:discussions}

The ultimate goal of our work is to formally define the notion of self-organization
for dynamic systems and to capture the conditions under which these systems are able to 
self-organize.  In the following we discuss two candidate models  to
capture the self-organization of dynamic systems: the
self-stabilization and the super-stabilization models~\cite{Dol00}.

In static systems, self-stabilization~\cite{Dol00} is an admirable
framework offering models and proof tools for systems that have the
ability to autonomously reach legitimate (legal) behavior despite
their initial configuration. 

\begin{definition}[Self-stabilization]
  \label{def:ss}
  Let ${\mathcal S}$ be a system and ${\mathcal L}$ be a predicate over
  its configurations defining the set of legitimate
  configurations. ${\mathcal S}$ is
  self-stabilizing for the specification $\mathcal{SP}$ if and only if  the following three properties hold:\\
  \hspace*{0.5cm}--- {\sf (convergence)} Any execution of ${\mathcal
    S}$ reaches
  a legitimate configuration that satisfies ${\mathcal L}$.\\
  \hspace*{0.5cm}--- {\sf (closure)} The set of legitimate configurations is closed.\\
  \hspace*{0.5cm}--- {\sf (correctness)} Any execution starting in a
  legitimate configuration satisfies the specification
  ${\mathcal{SP}}$.
\end{definition}

The following theorem enlightens
the limits of the self-stabilization definition in dynamic
systems. Specifically, Theorem~\ref{theo:imposibility}
shows that for each static demon ${\mathcal D}$
and for each self-stabilizing system ${\mathcal S}$ under ${\mathcal
  D}$ there exists a dynamic demon $\mathcal D^\prime$
such that ${\mathcal S}$ does not self-stabilize under $\mathcal D
\wedge \mathcal D^\prime$. In other words, the self-stabilization of a
system in dynamic settings is strongly related to the demon
definition. Note that similar arguments informally motivated the study of
super-stabilizing systems~\cite{DH95}.

\begin{theorem}
  \label{theo:imposibility}
  Let ${\mathcal S}$ be a system self-stabilizing under  demon
  ${\mathcal D}$. Let ${\mathcal L}$ be the set of the legitimate
  configurations of ${\mathcal S}$. There exists a dynamic demon
  ${\mathcal D}^\prime$
  such that $S$ does not self-stabilize under ${\mathcal D}^\prime$.
\end{theorem}

\begin{proof}
  The proof is based on the following intuitive idea which is  depicted in Figure \ref{fig:imposibility}.  For each
  execution of ${\mathcal S}$ an isomorphic execution
  that does not converge under $\mathcal D \wedge \mathcal D^\prime$ is constructed.
  Let $e$ be the execution of ${\mathcal S}$ under ${\mathcal D}$ and
  let ${\mathcal D}^\prime$ be defined as follows: each transition in
  ${\mathcal S}$ under ${\mathcal D}$ that leads to a legitimate
  configuration is replaced by a dynamic action 
  such that the new
  configuration is illegitimate.  Since ${\mathcal S}$ is
  self-stabilizing under ${\mathcal D}$ then any execution $e$ of
  ${\mathcal S}$ under ${\mathcal D}$ reaches a legitimate
  configuration. Consequently, in each execution $e$ there exists a
  configuration $c_1$ that precedes the legitimate configuration,
  $cl$. An arbitrary dynamic demon can always obtain from $c_1$ a new
  illegitimate configuration by replacing for example all the
  processes in $c_1$ with fresh processes erroneously initialized.

  \remove{ More formally, each $choice$, $ch$, in ${\mathcal D}$ such
    that the configuration $c_{l_1}$ is in ${\mathcal L}$, with $(c,
    ch, c_{l_1})$ transition in $e$ is replaced by a dynamic choice,
    $ch^d$ such that $c_{d_1}$ is illegitimate with $(c, ch^d,
    c_{d_1})$ transition in $e^d$, where $e^d$ is the execution of
    ${\mathcal S}$ under ${\mathcal D}^d$.  } Applying recurrently the
  replacement technique we can construct an infinite execution of
  ${\mathcal S}$ under ${\mathcal D} \wedge {\mathcal D}^\prime$ that
  never converges to a legitimate configuration. Hence, the system is
  not self-stabilizing under an arbitrary dynamic demon.
\end{proof}

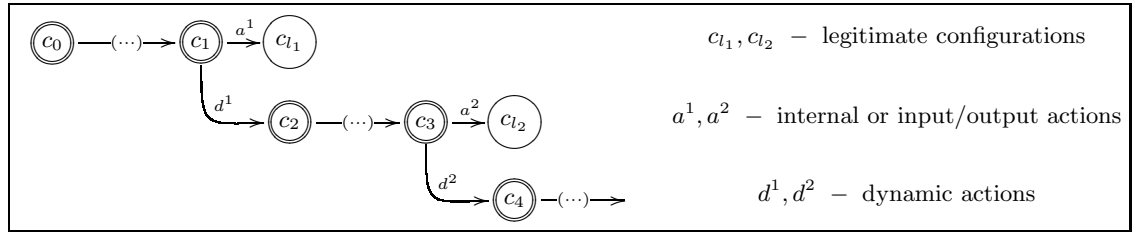
\begin{figure}\small
  \centering{ \fbox{ \xymatrix@-1pc{
        *++[o][F=]{c_0}\ar[rr]|{(\cdots)}&~~ & *++[o][F=]{c_1}
        \ar[r]^{a^1} \ar@(d,l)[dr]^{d^{1}} & *++[o][F]{c_{l_1}}
        &&&&&& *!{ c_{l_1} ,c_{l_2} ~-~ \textrm{legitimate configurations}} \\
        &&& *++[o][F=]{c_2}\ar[rr]|{(\cdots)} && *++[o][F=]{c_3}
        \ar[r]^{a^2} \ar@(d,l)[dr]^{d^{2}} & *++[o][F]{c_{l_2}}
        &&&*!{a^1,a^2
          ~-~\textrm{internal or input/output actions}}\\
        &&&&& &*++[o][F=]{c_{4}}\ar[rr]|{(\cdots)} & *{~~} &&
        *!{d^{1},d^{2}~-~\textrm{dynamic actions}} } }}
  \caption{Construction of a divergent execution}
  \label{fig:imposibility}
\end{figure}


Note that self-stabilization captures the self$^\ast$ aspects
of dynamic systems, in particular its self-organization, as long as \emph{i)}
the churn period is finite or \emph{ii)} static fragments are long enough so that the system stabilizes before a new period of churn occurs. 

Because of  the complexity of the new emergent systems (i.e.,  P2P,
sensors networks, robot networks) in terms
of both topology and data dynamicity,  coverage of these assumptions is low.  To go further in that direction, 
SuperStabilization has been proposed in the pioneering work of Dolev
and Herman~\cite{DH95}.  Superstabilization proposes a proof framework that
does not restrict the characterization of the system behavior to
static fragments only but rather extends it to dynamic periods.
%

Specifically, superstabilization proposes techniques for repairing a
self-stabi-lizing algorithm after C/D actions, such as the joining of fresh nodes  or after
topological changes due to the creation/removal of communication links.
communication links.
%

\begin{definition}[Superstabilization]
  A protocol or algorithm is (continuously) superstabilizing with respect
  to the class of topological changes $\Gamma$ if and only if  it is
  self-stabilizing and for every execution $e$ beginning in a
  legitimate state and containing only topology changes of type
  $\Gamma$ the passage predicate is verified for every configuration
  in $e$.
\end{definition}

Note that $\Gamma$, the class of  topological
changes  addressed by superstabilization,  are either single topology changes or multiple topology
changes provided that either each change is completely accommodated
before the next change occurs or changes are concurrent and each
set of concurrent changes is totally absorbed before the next
concurrent set occurs. This hypothesis is crucial for the safety  of
the system.  Otherwise the impossibility result shown  in Theorem~\ref{theo:imposibility} applies.


In~\cite{DH97,Her00} the authors  propose automatic mechanisms that transform
self-stabilizing algorithms into  superstabilizing ones. Colouring,
 spanning tree construction, depth first search, and mutual exclusion are among the studied algorithms. The basic idea of the transformation is the following one: when a topology change is detected the
system is frozen until the topology is repaired. During this repairing period, the
system verifies a transition predicate that captures the system
state just after the topology change occurred.

As will be clarified in the following sections, contrary to superstabilization we propose to characterize the notion of convergence for a dynamic system and propose sufficient conditions for
a system to self-organize according to the churn model.

It should be noted that any self-stabilizing or superstabilizing
system is also self-organizing under a finite dynamic scheduler or in
systems where each dynamic period is followed by long enough stability
periods.  However, a self-organizing system is not necessarily
self-stabilizing or superstabilizing.  We argue that our study is complementary
to the studies related to both self-stabilization and superstabilization, and opens an
alternative way for designing superstabilizing systems. That is,
superstabilization offers techniques for enforcing
self-stabilizing systems to keep their properties despite topology changes. However, most of
the systems are not self-stabilizing which restricts applicability of the superstabilization approach.  We
advocate that a superstabilizing system can be obtained from a
self-organizing system enriched with modules in charge of the system
stabilization.


\section{From Local Self-organization to Self-organization}
\label{sec:fromlocaltoself}

As discussed in the Introduction, self-organization is a global property that emerges from local interactions among nodes. To formally define what is a self-organized system, we first look at the system at a microscopic level, and then extend this study  to the macroscopic level, i.e., at the system as a whole. 

\subsection{Local Evaluation Criterion}
\label{sec:local}

Self-organization refers to the fact that the structure, the organization, or global properties of a system reveal themselves without any control from outside. This  emergence results only from internal constraints due to local interactions between its components or nodes. We model these local constraints as an objective function or evaluation criterion that nodes try to continuously maximize.

An objective function or evaluation criterion of a node  $p$ is an abstract
object defined on $p$'s local state 
and the local states of all its one/multi-hop neighbors.  Note that typically, the knowledge of a node is restricted to its one-hop
neighbors (as it is the case in sensors networks, cooperative robots networks, or P2P systems). However, it may happens that in some fine grained
atomicity models (as the read/write atomicity of Dolev~\cite{Dol00}) each
node maintains copies of both its local state and the ones of its
neighbors at more than one hop from itself. Therefore in these models it makes sense for a node  to
compute an evaluation criterion  based on both its local state and the
ones of its multi-hop neighbors. 
%


 In the following, $\gamma_{p,c}$  will denote the evaluation criterion  at node 
 $p$ in configuration $c$. $\gamma_{p,c}$  is  a $[0,1]$ function defined on 
 a co-domain CD equipped with a partial order relation ${\mathcal R}$. 
 The relation $\mathcal R$ is the $\leq$ relation on real
numbers. Function  $\gamma_{p,c}$ represents the aggregate of $\gamma_{p,c}(q)$ 
for all neighbors $q$ of $p$ in configuration $c$. 

In order to define the local self-organization, we introduce the
notion of \emph{stable configurations}. Informally, a
configuration~$c$ is $p$-\emph{stable} for a given evaluation
criterion in the neighborhood of node~$p$ if the local criterion
reached a local maximum in $c$.

\begin{definition}[$p$-stable configuration]
  Let $c$ be a configuration of a system~$\mathcal{S}$,  $p$ be a
node, $\gamma_{p,c}$ be the local criterion of $p$ in configuration 
  $c$ and $\preceq$ be a partially order relation on the codomain of
  $\gamma_{p,c}$.  Configuration~$c$ is $p$-\emph{stable} for $\gamma_{p,c}$
  if, for any configuration~$c^\prime$ reached from $c$ after one
  action executed by $p$, $\gamma_{p,c^\prime} \preceq \gamma_{p,c}$.
\end{definition}

\begin{definition}[Local self-organization]
  Let $\mathcal{S}$ be a system, $p$ be a process, and $\gamma_{p}$ be the local
  criterion of $p$.  $\mathcal{S}$ is locally self-organizing for
  $\gamma_p$ if $\mathcal{S}$ eventually reaches a $p$-stable
  configuration.  
\end{definition}

In P2P systems local self-organization should force nodes to have as 
logical  neighbors nodes that  improve a sought  evaluation
criterion.  Module~\ref{alg:LSA} executed by node~$p$, referred
to as LSA in the sequel, proposes a local self-organizing generic
algorithm for an arbitrary  criterion~$\cC_p$.  Note that
existing DHT-based peer-to-peer systems  execute similar algorithms to ensure
self-organization with respect to specific criteria
(e.g.,~geographical proximity) as shown in the sequel. The nice property of our generic
algorithm is its adaptability to unstructured networks.

LSA is based on a greedy technique, which reveals to be a well adapted
technique for function optimization.  Its principle follows the here
above intuition: Let $q$ such that $q \in\mathcal{N}^{\cC_p}(p)$, and
$r$ such that $r \in\mathcal{N}^{\cC_p}(q)$ but $r
\not\in\mathcal{N}^{\cC_p}(p)$, where $\mathcal{N}^{\cC_p}(p)$ and
$\mathcal{N}^{\cC_p}(q)$ are the logical neighborhoods of $p$ and $q$
respectively with respect to criterion~$\cC_p$.  If $p$ notices in configuration $c$ that
$r$ improves the evaluation criterion previously computed for $q$,
then $p$ replaces $q$ by $r$ in $\mathcal{N}^{\cC_p}(p)$. Inputs of this
algorithm are the evaluation criterion $\cC_p$ and the set of $p$'s
neighbors for $\cC_p$, that is $\mathcal{N}^{\cC_p}(p)$.  The output is
the updated view of $\mathcal{N}^{\cC_p}(p)$. Given a criterion~$\cC_p$, a
$p$-stable configuration $c$, in this context, is a configuration where for
any neighbor~$q$ of $p$, there is no neighbor~$r$ of $q$ ($r \neq p$)
that improves $\cC_p$, formally $\forall q \in \mathcal{N}^{\cC_p}(p),
\forall r \in \mathcal{N}^{\cC_p}(q) \setminus \mathcal{N}^{\cC_p}(p),~
{\cC_{p,c}}(r) \leq {\cC_{p,c}}(q)$.

Note that, because of the partial view that a node  has on the
global state of the system (due to the scalability and dynamism of the
system), only a heuristic algorithm can be found under these
assumptions.

\begin{Module}
  \textbf{Inputs:}\\
  \hspace*{0.5cm}
  $\cC_p$: the evaluation criterion used by $p$;\\

  \hspace*{0.5cm} ${\mathcal N}^{\cC_p}(p)$: $p$ neighbors for the
  evaluation criterion $\cC$;
  \begin{tabbing}
    \textbf{Actions:}\\
    \hspace*{.5cm} \=
    $\mathcal{R}~:~$\=\textsf{if} $\exists q \in {\mathcal N}^{\cC_p}(p), \exists r \in \mathcal{N}^{\cC_{p}}(q)\setminus \mathcal{N}^{\cC_{p}}(p), {\cC_{p,c}}(q) < {\cC_{p,c}}(r) $\\
    \> \> \textsf{then}~~$\mathcal{N}^{\cC_p}(p)=\mathcal{N}^{\cC_p}(p) \bigcup \{r_\mathit{max}\} \setminus {q}$;\\
    \>\>\hspace*{1cm} 
    where $r_\mathit{max} \in {\mathcal N}^{\cC_p}(q)$,
    ${\cC_{p,c}}(r_\mathit{max})= \max_{r^\prime\in {\mathcal N}^{\cC_p}(q),
      {\cC_{p,c}}(q) \leq {\cC_{p,c}}(r^\prime)}({\cC_{p,c}}(r^\prime))$
  \end{tabbing}
  \caption{Local Self-Organization Algorithm for Criteria $\cC$
    Executed by $p$ in configuration $c$ (LSA)}
  \label{alg:LSA}
\end{Module}

\begin{theorem}[Local Self-organization of LSA]
  \label{lsa-theorem}
  Let $\mathcal{S}$ be a system and $\gamma_p$ be an objective function of node $p$. If $p$ executes the LSA algorithm with $\gamma_p$,
  then $\mathcal{S}$ is a locally self-organizing system for $\gamma_p$.
\end{theorem}

\begin{proof}
  Let $p$ be a node  in the system executing the LSA algorithm.
  Assume that $\mathcal{S}$ does not locally self-organize in the
  neighborhood of $p$. That is, there is an execution of
  $\mathcal{S}$, say $e$, that does not include a $p$-stable
  configuration.
  
  Assume first that $e$ is a static execution (i.e.,~no
  C/D action is executed during $e$). Let $c$ be
  the first configuration in $e$. By assumption of the proof, $c$ is
  not $p$-stable. Thus there is a neighbor of $p$, say $q$, that has
  itself a neighbor improving the evaluation criterion. Hence, rule
  ${\mathcal R}$ (Module \ref{alg:LSA}) can be applied which makes $r$
  replacing $q$ in the neighbors table of $p$. By applying the
  assumption of the proof again, the obtained configuration is not
  stable, hence there is at least one neighbor of $p$ which has a
  neighbor which improves the evaluation criteria. Since the
  evaluation criteria is bounded and since the replacement of a
  neighbor is done only if there is a neighbor at distance $2$ which
  strictly improves the evaluation criterion, then either the system
  converges to a configuration $c_\mathit{end}$ where the evaluation
  criterion reaches its maximum for some neighbors of $p$, or the
  evaluation criterion cannot be improved.

  In other words, for each node $q$ neighbor of $p$ we can exhibit a
  finite maximal string:
  $$\gamma_{p,c_0}(q_0) < \gamma_{p,c_1}(q_1) < \ldots < \gamma_{p,\mathit{end}}(q_m)$$ 
    where $q_0$ is the node $q$ and $q_i$,
  $i=\overline{1,m}$ are the nodes which will successively replace the
  initial node~$q$.  Let $c_\mathit{end}$ be the configuration where
  the node~$q_m$ is added to the neighbors table of $p$.  In
  $c_\mathit{end}$ the value of $\gamma_{p,c_\mathit{end}}(q_m)$ is maximal hence, either
  $c_\mathit{end}$ is stable, or no neighbor of $q_m$ improves the
  evaluation criterion.  Thus $c_\mathit{end}$ is stable. Consequently,
  there exists a configuration in $e$, namely $c_\mathit{end}$, which
  is $p$-stable.
 
  Assume now that the execution $e$ is dynamic, hence the system size
  and topology may be modified by nodes connection and disconnection.
  Assume that node $p$ joins the system. This case is similar to the
  previous one, where $p$ executes rule ${\mathcal R}$ of Module
  \ref{alg:LSA} a finite number of times until it reaches a $p$-stable
  configuration.
  
  Now, let us study the case where the system is in a $p$-stable
  configuration and, due to the connection of some node $r$,  $p$'s 
  neighborhood changes.  That is $r$
  appears in $p$'s 
  neighborhood. Suppose that  $r$ improves $\gamma_p$.  Then $p$ applies  rule ${\mathcal R}$.  As previously shown, the system
  reaches in a finite number of steps a $p$-stable configuration.
  \qed
\end{proof}

\subsection{Global Evaluation Criterion}
\label{sec:global_evaluation_criterion}

We now introduce the notion of \emph{global evaluation criterion}, denoted
in the following $\gamma$.  The global evaluation criterion evaluates
the global organization of the system at a given configuration. More
precisely, the global evaluation criterion is the aggregate of all
local criteria.  For instance, if the evaluation criterion is the
logical proximity (i.e., the closer a process, the higher the
evaluation criterion), then optimizing the global evaluation criterion
$\gamma$ will result in all processes being connected to nearby
processes.

Let $\gamma_p$ be the local criterion of  process $p$, and $\preceq$ be a partially order relation on the codomain of
  $\gamma_{p}$.  In the
sequel we focus only on global evaluation criteria $\gamma$ that exhibit the following property:
\begin{description}
\item[Global evaluation criterion property]
\begin{eqnarray*}
  \forall f, \forall c_1, c_2 
  \in f, ~\gamma(c_1) \prec \gamma(c_2) \mbox{~if~} 
  \exists p, ~\gamma_p(c_1,\cC_p) \prec \gamma_p(c_2,\cC_p)
  \mbox{~and~} \\
  \forall t \neq p,~ \gamma_t(c_1, \cC_t) \preceq \gamma_t(c_2, \cC_t)
\end{eqnarray*}
\end{description}

Intuitively, the increase of the value of a local criterion will
determine the increase of the global criterion if the other local
criteria increase their values or remain constant.  An example of
criterion that meets such requirements is the union/intersection of
$[0,1]$-valued local criteria. Namely, $\gamma$ is the sum of a local
aggregation criterion $\gamma$: $\displaystyle \gamma(c)=\sum_{p\in
  \mathcal{S}}~\gamma_p(c)$.
  
  We now define the notion of self-organization in a dynamic system: 
  
  \begin{definition}[Self-organization] 
 Let $\mathcal{S}$ be a system, $p$ be a process, and $\gamma_{p}$ be the local
  criterion of $p$.   System $\mathcal{S}$ is self-organizing if $\forall
  p \in \mathcal{S}$, $\mathcal{S}$ is locally self-organizing for
  $\gamma_p$.
  \end{definition}

\section{From Weak Self-organization to Strong Self-organization}
\label{sec:different-level}

 
The next three sections present the different forms of self-organization a dynamic system should exhibit. These different forms of self-organization are strongly related to the churn model the system  is subject to. In its weakest form, self-organization is only required during static fragments, i.e., in absence of any dynamic actions. Clearly, this imposes no restriction on the churn model, and thus weakly self-organized systems should be able to tolerate an arbitrary dynamic demon. In its strongest form, self-organization should guarantee that the system as a whole is  able to maintain or even progressively increase its global knowledge at any time during its execution (that is in both static fragments and instability periods). Contrary to its weakest form, this limits the strength of the demon. In particular this requires that in any configuration there exists some part of the system (not necessarily the same one) for which the knowledge  remain constant or even increases. Therefore,  this imposes that during instability periods, some part of the system does not get  affected by churn. This behavior is encapsulated in the kernel-based dynamic demon. 
Finally, in between these two extreme forms, self-organization should guarantee that infinitely often the entropy of the system should reduce, that is period of times during which the objective function at each node should increase or at least remain constant. This requires to limit the number of C/D dynamic actions that may occur during any system execution. 
All these different forms of self-organization are described according to a safety and liveness property. Safety  must be preserved at all times, while liveness can be temporarily hindered by instability, but eventually when the system enters static fragments some progress must be accomplish.

\subsection{Weak Self-Organization}
\label{sec:global}

The weak self-organization is defined in terms of two properties.  The
\emph{weak liveness} property says that for each static
fragment~$f_i$, either (1)~$f_i$ is stable, or (2)~there exists some
fragment~$f_j$, in the future of $f_i$, during which the global
evaluation criteria strictly improves (see
Fig.~\ref{fig:weak-liveness}).  The \emph{safety} property requires
that the global evaluation criteria never decreases during a static
fragment.  Formally, we have:

\begin{definition}[Weak Self-Organization]
\label{def:weak-self-org}
  Let $\mathcal{S}$ be a system and $\gamma$ be a global evaluation
  criterion defined on the configurations of $\mathcal{S}$.  A system
  is weakly self-organizing for $\gamma$ if the following two
  properties hold (recall that $(f_0, \ldots, f_i, \ldots)$ stand for
  static fragments):
      
    \begin{property}[Safety Property]\label{prop:sp}
    \begin{eqnarray*}
     \forall e=(f_0, \ldots, f, \ldots), \forall f \in e: \gamma(\debut(f)) \preceq \gamma(\fin(f)).
  \end{eqnarray*}
  \end{property}
  
  \begin{property}[Weak Liveness Property]\label{prop:wlp}
    \begin{eqnarray*}
      \forall e\!=\!(f_0, \ldots, f_i, \ldots, f_j, \ldots), \forall f_i\!\in\!e,
      \exists f_j\!\in\!e, && j\!\geq\!i:~  \gamma(\debut(f_j))  \prec \gamma(\fin(f_j))\\
      && \textsf{or} ~ \forall p \in \mathcal{S}, \debut(f_j)
      ~\mbox{is \ensuremath{p}-stable}.
    \end{eqnarray*}
    \end{property}

\end{definition}
\noindent

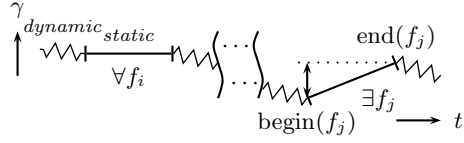
\begin{figure}[t]
  \centering
  \begin{pspicture}(0,0)(\pu,0.25\pu)
    \psline{->}(0,.1\pu)(0,.2\pu) \uput[u](0,.2\pu){$\gamma$}
    \psline{->}(.85\pu,0)(.95\pu,0) \uput[r](.95\pu,0){$t$}

    \pszigzag[coilarm=0mm,coilwidth=2mm,linewidth=.5pt]{-}(.05\pu,.16\pu)(.15\pu,.15\pu)
    \psline{|-|}(.15\pu,.15\pu)(.35\pu,.15\pu)
    \uput[d](.25\pu,.15\pu){$\forall f_i$}
    \pszigzag[coilarm=0mm,coilwidth=2mm,linewidth=.5pt]{-}(.35\pu,.15\pu)(.45\pu,.13\pu)
    
    \pscurve{-}(.45\pu,.2\pu)(.44\pu,.17\pu)(.45\pu,.13\pu)(.46\pu,.09\pu)(.45\pu,.05\pu)
    \uput[u](.49\pu,.125\pu){$\cdots$}
    \uput[d](.51\pu,.125\pu){$\cdots$}
    \pscurve{-}(.54\pu,.2\pu)(.53\pu,.17\pu)(.54\pu,.13\pu)(.55\pu,.09\pu)(.54\pu,.05\pu)

    \pszigzag[coilarm=0mm,coilwidth=2mm,linewidth=.5pt]{-}(.55\pu,.07\pu)(.65\pu,.05\pu)
    \psline{|-|}(.65\pu,.05\pu)(.85\pu,.13\pu)
    \uput[dr](.75\pu,.09\pu){$\exists f_j$}
    \pszigzag[coilarm=0mm,coilwidth=2mm,linewidth=.5pt]{-}(.85\pu,.13\pu)(.95\pu,.1\pu)

    \psline[linestyle=dotted]{-}(.85\pu,.13\pu)(.63\pu,.13\pu)
    \psline{<->}(.65\pu,.05\pu)(.65\pu,.13\pu)

    \uput[u](.1\pu,.16\pu){\scriptsize \emph{dynamic}}
    \uput[u](.25\pu,.15\pu){\scriptsize \emph{static}}
    \uput[d](.65\pu,.05\pu){\footnotesize $\debut(f_j)$}
    \uput[u](.85\pu,.13\pu){\footnotesize $\fin(f_j)$}
  \end{pspicture}
  \caption{Illustration of the Weak liveness property.}
  \label{fig:weak-liveness}
\end{figure}
The following theorem gives a sufficient  condition to build a
weakly self-organizing system:

\begin{theorem}[Weak Self-organization]
  \label{theo:weak}
  Let $\mathcal{S}$ be a system and $\gamma$ be a  global
  objective function. System $\mathcal{S}$ is weakly self-organizing for
  $\gamma$ if for any node  $p$ system $\mathcal{S}$ locally
  self-organizes in $p$'s neighborhood.
\end{theorem}

\begin{proof}
  Let $e$ be an execution of system $\mathcal{S}$.
  
  \begin{description}
  \item[Safety] Let $f$ be a static fragment in $e$.  By assumption 
    $\mathcal{S}$ is locally self-organizing for $\gamma_p$ for any node $p\in S$. Thus we  have two situations:
     either  $p$ executes some particular 
    actions in $f$ that makes $\gamma_p(\debut(f)) \prec \gamma_p(\fin(f))$ true
    or  $p$ does not execute any action and in that case,
    $\gamma_p(\debut(f)) \preceq \gamma_p(\fin(f))$. Overall,
   by definition of a global objective function this leads to  $\gamma(\debut(f)) \preceq \gamma(\fin(f))$. This proves Property~\ref{prop:sp}.

  \item[Weak liveness] Let $p$ be some node in $S$. Let $f_i$ be an
    arbitrary static fragment in $e$.  By assumption $\mathcal{S}$ is
    locally self-organizing. Thus  there is a static fragment $f_j,~ i\leq j$
    in $e$ such that $p$ executes a particular action in $f_j$
    that makes $\gamma_p(\debut(f_j)) \prec \gamma_p(\fin(f_j))$.  Overall, for any
    $f_i$ there is a fragment $f_j$ such that $ \gamma(\debut(f_j))
    \prec \gamma(\fin(f_j)) $. This completes the proof of Property~\ref{prop:wlp}, and ends the proof of the theorem.
  \end{description}
  \qed
\end{proof}

\subsection{Case Study of Weak Self-Organization: CAN}
\label{sec:can_case_study}

In this section, we prove the self-organization of CAN.
\label{sec:can}
CAN~\cite{CAN-sigcomm} is a scalable content-addressable network, the
principle of which is to use a single namespace---the $d$-dimensional
torus $[0,1]^d$---for both data and nodes.  Data and CAN nodes are
assigned unique names within this namespace, and each node is
responsible for a volume surrounding its identifier in the torus. The
distributed algorithm executed on a node arrival or departure ensures
that the complete torus volume is partitioned between all
participating CAN nodes.

These algorithms are crucial for the self-organization of the system,
since the topology of CAN changes only when nodes enter or leave the
system. In the following, we show how these protocols fit into our
self-organization framework.  
\noindent
Consider the following local criterion at node $p$:
\begin{eqnarray}
\gamma_p^\mathit{CAN}(q)\stackrel{\rm def}{=}\frac{1}{1+\mathit{dist}(p,q)},
\end{eqnarray}
where $\mathit{dist}$ is the Cartesian distance in the torus.

\begin{lemma}[CAN]
  CAN is weakly self-organized with respect to 
  $\gamma^{CAN}_p$ objective function.
\end{lemma}

\begin{proof} 
  We show that protocols for node insertion and node
  removal in CAN perform actions that keep  the system in a $p$-stable
  configuration. 

  \begin{description}
  \item[Node Removal.]%
    When a node leaves the system, its previous neighbors' evaluation
    criteria decrease, since the distance to the leaving node is infinite. As we are only concerned with fragments in which
    no disconnection can occur, let us consider actions taken by the
    protocol subsequently to the leaving of node $p$.  Just after the
    departure of $p$, every neighbor of $p$ saw a decrease in its
    evaluation function, and starts to look for a new neighbor. The
    algorithm used by CAN~\cite{CAN-thesis,CAN-sigcomm} is designed in
    such a way that the newly chosen neighbor is optimal with respect
    to the Cartesian distance. Hence, in the fragment following
    the leaving of $p$, the criterion for every neighbor of $p$
    increases.  Once each node which used to be a neighbor of the leaving node has completed  the protocol, then the topology of CAN does not change
    unless a connection or another disconnection occur.  Hence, the
    departure protocol leaves the system in a $p$-stable
    configuration.
  
  \item[Node Insertion.]%
    The insertion of a node is a two-step operation.  In the first
    step, the node $p$ that wants to join the system computes an
    $\mathit{id}$, which is a point in the $d$-torus, and then gets the IP
    address of some CAN node $q_0$. The second step is the actual
    insertion: (1)~$q_0$ sends a message to the node $q_1$ responsible
    for the volume containing the $\mathit{id}$ computed by $p$, then
    (2)~$p$ contacts $q_1$ which, in turn, contacts its neighbors and
    splits its volume in order to maximize the uniform distribution of
    nodes within the torus, and finally (3)~$p$ enters the system with
    a volume defined by its $\mathit{id}$ and by $q_1$ and its
    neighbors.

  \end{description}

  The key point here is that, for any node $r$ in the torus, when a
  new node $p$ is inserted in CAN, it becomes a neighbor of $r$ only
  if $p$ is closer to $r$ than one of $r$'s previous neighbors. Hence,
  the Cartesian distance from $r$ to its neighbors is either the same
  or reduced, when compared to the situation before the insertion: the
  evaluation criterion for every node in the system is improved by an
  insertion. Thus by Theorem~\ref{theo:weak}, this makes CAN weakly self-organized.
  \qed
\end{proof}

\subsection{Case Study of Weak Self-Organization: Pastry}
\label{sec:pastry}

Pastry~\cite{pastry} is a general substrate for the construction of peer-to-peer
Internet applications like global file sharing, file storage, group
communication and naming systems. Each node and key in the Pastry
peer-to-peer overlay network are assigned a 128-bit node identifier
($\mathit{nodeId}$). The nodeId is used to indicate a node's position
in a circular $\mathit{nodeId}$ space, which ranges from $0$ to
$2^{128}-1$. The $\mathit{nodeId}$ is assigned randomly when a node
joins the system. The distributed algorithms run upon node arrivals and
departures rely on local structures (i.e.,~routing table, leaf set, and
neighborhood set) to ensure that nodes are correctly positioned on the
circular $\mathit{nodeId}$ space. In the following we study Pastry self-organization by characterizing the local criterion used by any node $p$ to update the routing table, the leaf set and the neighborhood set. These local criterion are respectively denoted by $\gamma_p^\mathit{Pastry,routing}(.)$, $\gamma_p^\mathit{Pastry,neighbor}(.)$, and $\gamma_p^\mathit{Pastry,leaf}(.)$. 

\paragraph{Routing Table}
\label{sec:routing-table}
The routing table $R(p)$ of any node $p$ contains the $\mathit{nodeId}$s of the
nodes that share a common prefix with $p$ node.  More precisely, it is
organized into $\lceil \log_{2^b}N \rceil$ rows with $2^b-1$ entries
each.  Each of the $2^b-1$ entries at row~$\ell$ refers to a node whose
nodeId shares the first $\ell$ digits of $p$ nodeId, but whose
$\ell+1$th digit has one of the $2^b-1$ possible values other than the
$\ell+1$th digit of $p$ nodeId. If several nodes have the
appropriate prefix, the one which is the closest according to 
distance metric $\mathit{dist}$ to node~$p$ is chosen.
\noindent
Consider the following local criterion  at node $p$:
\begin{eqnarray}
  \gamma_p^\mathit{Pastry,routing}(q)\stackrel{\rm def}{=} \frac{f(i,j,k,n)}{\mathit{dist}(k,n)},
\end{eqnarray}
where $\mathit{dist}$ is a  a scalar proximity metric, and $f$ is defined in Relation~(\ref{rel:f}). 
Let $p$ and $q$ be two nodes such that $p=p_0 p_1 \ldots p_{\lfloor
  \log_{2^b}N-1 \rfloor}$, and $q=q_0 q_1 \ldots q_{\lfloor
  \log_{2^b}N-1 \rfloor}$
\begin{eqnarray}\label{rel:f}
  f(i,j,p,q) = 
  \left\{ \begin{array}{ll}
      1 & \mbox{ \ if $\bigwedge_{l=0}^{i-1} (p_l=q_l) \wedge (j=q_i)$ is true }\\
      0 & \mbox{ \ otherwise.}
    \end{array}    
  \right.
\end{eqnarray}

\begin{lemma}[Pastry, Routing]
   Pastry 
  is weakly
  self-organized with respect to $\gamma_p^\mathit{Pastry,routing}$
  objective function.
\end{lemma}

\begin{proof} We show that the algorithms used for node arrival and node
removal perform actions that leave the system in a $p$-stable
configuration. 

\begin{description}
\item[Node Removal.] Suppose that some node $q$ leaves the system. Then,    the evaluation
  criterion $\gamma_p^\mathit{Pastry,routing}(q)$ of node $p$ is set to zero. Thus, node $p$  tries to update its failed routing table entry
  $R_\ell^d(p)$. This is achieved  by asking the node pointed by another entry
  $R_\ell^i(p)$ with $i \neq d$ for an appropriate reference.  If none of the nodes of the row have a
  pointer on a lived node, $p$ asks for nodes referred to row $l+1$,
  thereby casting a wider net. With high probability, this procedure
  eventually finds an appropriate node if such a node exits in the system.
  Thus, $p$ updates its evaluation criterion. By  
  Theorem~\ref{theo:weak}, the system reaches a $p$-stable
  configuration.
  
\item[Node Arrival.] Suppose that some node $p$ joins the system. Then by construction, $p$ initializes
  its state tables and informs other nodes of its presence. It is
  assumed that $p$  knows initially about a nearby Pastry node
  $q_0$, according to the proximity metric. The insertion algorithm
  enables the new node $p$ to find a node $q_\ell$ whose nodeId is
  numerically closest to $p$'s nodeId. Node $p$ obtains the $\ell$th row
  of its routing table from the node encountered along the path from
  node $q_0$ to node $q_\ell$ whose nodeId matches $p$ in the first $\ell-1$
  digits. Assuming the triangle inequality holds in the proximity
  space, an easy induction leads to the fact that the entries of the
  $\ell$th row of $q_\ell$ routing table should be close to $p$. Thus  
  $\gamma_p^\mathit{Pastry,routing}$ progressively
  increases.  We now show how the routing tables of the affected nodes
  are updated to take into account the arrival of node $p$. Once node
  $p$ reaches node $q_\ell$,  $p$ sends a copy of the $\ell$th row of its
  routing table to all the nodes that are pointed to in that row $\ell$. Upon
  receipt of such a row, a node checks whether node $p$ or one of the
  entries of that row are nearer than the existing entry (if one
  exists). In that case, the node replaces the existing entry by the
  new one.  Thus the evaluation criterion for all these nodes is
  improved by an insertion.  Note that even if the evaluation
  criterion is improved, due to the heuristic nature of the node
  arrival and the fact that the practical proximity metrics do not
  satisfy the triangle inequality, one cannot ensure that the routing
  table entries produced by a node arrival are the closest to the
  local node. To prevent that such an issue  leads to a deterioration of
  the locality properties, routing table maintenance is periodically triggered.
\end{description}
\qed
\end{proof}

\paragraph{Neighborhood Set}
\label{sec:neigb}
The neighborhood set contains the nodeIds and IP addresses of the
$|M|$ nodes that are closest according to the proximity metric of 
node $p$.  The typical value for the size of this set is $2^b$
or $2*2^b$.
\noindent
Consider the following local  criterion at node $p$:
\begin{eqnarray}
\gamma_p^\mathit{Pastry,neighbor}(q)\stackrel{\rm def}{=}\frac{1}{1+\mathit{dist}(p,q)}
\end{eqnarray}
where $\mathit{dist}$ is a scalar  proximity metric.

\begin{lemma}[Pastry, neighbor]
  Pastry  
   is weakly
  self-organized with respect to  $\gamma^\mathit{Pastry,neighbor}_p$
  objective function.
\end{lemma}

\begin{proof} 
  We show that the algorithms used for node arrival and node removal
  perform actions that leave the system in a $p$-stable configuration.
  \begin{description}
  \item[Node Removal.] The neighborhood set is periodically checked. If
    some node $q$ does not respond to requests of node $p$ (leading to 
    decrease $\gamma_p^\mathit{Pastry,neighbor}(q)$), $p$ asks other
    members of this set to give it back their neighbor sets. Node $p$
    checks the distance of each of the new nodes, and inserts the node
    that is the closest to it. Thus, for each neighbor $q$ of $p$,  $\gamma_q^\mathit{Pastry,neighbor}$ increases.  Theorem~\ref{theo:weak}
    implies that the system reaches a $p$-stable configuration.
  \item[Node Arrival.]  As previously said, when a new node $p$ joins
    the system it contacts a node close to it. This node gives its
    neighborhood set to $p$ so that $p$ is able to initialize its own
    neighborhood set, increasing thus its objective function. Clearly, since $q_0$ is in the proximity of $p$, both
    have a close neighborhood. Then $p$ proceeds as explained here
    above, and transmits a copy of its neighborhood set to each of the
    nodes found in this set. Those nodes update their own set if $p$
    is an appropriate node. Thus the objective function  for all
    these nodes is improved by an insertion.
  \end{description} 
  This completes the proof of the lemma.
  \qed
\end{proof}

\paragraph{Leaf Set}
\label{sec:leaf-set}
The leaf set is the set of nodes with the $L/2$ numerically
closest larger nodeIds, and the $L/2$ nodes with numerically closest smaller nodeIds,
relative to $p$ nodeId. Typical values for $L$ are $2^b$ or $2*2^b$.
\noindent
Consider the following local criterion at node $p$:
\begin{eqnarray}
\gamma_p^\mathit{Pastry,leaf}(q)\stackrel{\rm def}{=}\frac{1}{1+\mathit{dist}_\mathit{nodeId}(p,q)}
\end{eqnarray}
where $\mathit{dist}$ is the proximity metric in the nodeId space.

\begin{lemma}[Pastry, leaf]
  The Pastry system with respect to the leaf set is a weak
  self-organizing system using the $\gamma^\mathit{Pastry,leaf}_p$
  criterion.
\end{lemma}

\begin{proof} 
  We show that the algorithms used for node arrival and node removal
  perform actions that leave the system in a $p$-stable configuration.
  \begin{description}
  \item[Node Removal.]  Suppose that some node $q$ leaves or fails. Then for all nodes $p$ that are neighbors of $q$, we have  $\gamma_p^\mathit{Pastry,leaf}(q)$ decreased. To
    replace a node in the leaf set, $p$ contacts the node with the largest index with respect to 
     $q$, and asks that node for its leaf
    set.  This leaf set may overlap $p$ leaf set and may contain nodes which are not in $p$ leaf set.  Among these new nodes, $p$ chooses
    the appropriate one to insert into its leaf set. Thus, unless
    $L/2$ fail simultaneously, then $p$ evaluation criterion
    increases.
  \item[Node Arrival.] The same argument as above applies, except
    that the node that gives to $p$ its leaf set is no more node $q_0$
    but $q_n$, since $q_n$ has the closest existing nodeId with respect to
    $p$. Then $p$ proceeds as explained above, and transmits a
    copy of its leaf set to each of the nodes found in this set. Those
    nodes update their own set if $p$ is an appropriate node. Thus the
    objective function  for all these nodes is improved by an
    insertion.
  \end{description} \qed
\end{proof}

\subsection{Self-Organization}
\label{sec:self-org}
As previously described, the weak self-organization definition applies to static fragments.
Nothing is guaranteed during dynamic ones (i.e.,~fragments in which
connections / disconnections occur or data are modified).  For
example, Pastry self-organization protocol may cause the creation of
multiple, isolated Pastry overlay networks during periods of IP
routing failures. Because Pastry relies almost exclusively on
information exchanged within the overlay network to self-organize,
such isolated overlays may persist after full IP connectivity resumes
\cite{DR01}. 

A self-organized system should be able to infinitely often increase its knowledge even in presence of churn. We characterize this  gradual enrichment through the safety property as defined above and a liveness
property that  says that either (1)~infinitely often, there
are static fragments during which the knowledge of the system enriches
(cf. Fig.~\ref{fig:liveness}), or (2)~all the processes have reached a
stable state.

\begin{figure}
  \centering
  \begin{pspicture}(0,0)(\pu,0.25\pu)
    \psline{->}(0,.1\pu)(0,.2\pu) \uput[u](0,.2\pu){$\gamma$}
    \psline{->}(.85\pu,0)(.95\pu,0) \uput[r](.95\pu,0){$t$}

    \pszigzag[coilarm=0mm,coilwidth=2mm,linewidth=.5pt]{-}(.05\pu,.16\pu)(.15\pu,.15\pu)
    \psline{|-|}(.15\pu,.15\pu)(.35\pu,.17\pu)
    \uput[u](.25\pu,.16\pu){$\forall f_i$}
    \pszigzag[coilarm=0mm,coilwidth=2mm,linewidth=.5pt]{-}(.35\pu,.17\pu)(.45\pu,.15\pu)

    \pscurve{-}(.45\pu,.2\pu)(.44\pu,.17\pu)(.45\pu,.13\pu)(.46\pu,.09\pu)(.45\pu,.05\pu)
    \uput[u](.49\pu,.125\pu){$\cdots$}
    \uput[d](.51\pu,.125\pu){$\cdots$}
    \pscurve{-}(.54\pu,.2\pu)(.53\pu,.17\pu)(.54\pu,.13\pu)(.55\pu,.09\pu)(.54\pu,.05\pu)

    \pszigzag[coilarm=0mm,coilwidth=2mm,linewidth=.5pt]{-}(.55\pu,.07\pu)(.65\pu,.05\pu)
    \psline{|-|}(.65\pu,.05\pu)(.85\pu,.23\pu)
    \uput[dr](.75\pu,.14\pu){$\exists f_j$}
    \pszigzag[coilarm=0mm,coilwidth=2mm,linewidth=.5pt]{-}(.85\pu,.23\pu)(.95\pu,.2\pu)

    \psline[linestyle=dotted]{-}(.85\pu,.23\pu)(.33\pu,.23\pu)
    \psline{<->}(.35\pu,.17\pu)(.35\pu,.23\pu)

    \uput[u](.85\pu,.23\pu){\footnotesize $\fin(f_j)$}
    \uput[d](.35\pu,.17\pu){\footnotesize $\fin(f_i)$}
  \end{pspicture}
  \caption{Illustration of the liveness property.}
  \label{fig:liveness}
\end{figure}
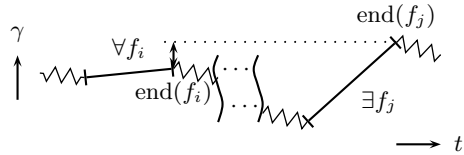

\begin{definition}[Self-Organization]\label{def:self-org}
  Let $\mathcal{S}$ be a system and $\gamma$ be a global evaluation
  criterion defined on the configurations of $\mathcal{S}$.  A system
  is \emph{self-organizing} for $\gamma$ if both safety as defined in Property~\ref{prop:sp} and liveness hold, with the liveness property defined in Property~\ref{prop:lp}.

  \begin{property}[Liveness Property]\label{prop:lp}
    \begin{eqnarray*}
      \forall e\!=\!(f_0, \ldots, f_i, \ldots, f_j, \ldots), 
      \forall f_i\!\in\!e,
      \exists f_j\!\in\!e, && j\!\geq\!i: 
      \gamma(\fin(f_i))\!\prec\!\gamma(\fin(f_j)) \\
      && \mbox{\textsf{or} } \forall p\!\in\!\mathcal{S}, \debut(f_j)
      ~\mbox{is \ensuremath{p}-stable}
    \end{eqnarray*}
  \end{property}
  
\end{definition}

\begin{theorem}[Self-organization]
  \label{theo:selforg}
  Let $\mathcal{S}$ be a locally self-organizing system.  If, for any
  execution $e = (f_0, \ldots, f_i, \ldots, f_j, \ldots)$ of
  $\mathcal{S}$ and for all static fragments $f_i$ and $f_{i+1}$ in execution $e$ we have 
  $\gamma(\fin(f_i)) \leq \gamma(\debut(f_{i+1}))$ then, $\mathcal{S}$
  is self-organizing.
\end{theorem}

\begin{proof}
  From the local self-organization of $\mathcal{S}$, for all $f_i$,
  it exists $f_j$ such that $\gamma(\debut(f_j)) \prec
  \gamma(\fin(f_j))$.  Using the hypothesis, we have $$\gamma(\debut(f_i))
  \preceq \gamma(\fin(f_i)) \preceq \gamma(\debut(f_{i+1})) \ldots
  \linebreak \preceq \gamma(\debut(f_j)) \prec
  \gamma(\fin(f_j)).$$ Thus $\mathcal{S}$ is self-organizing.  \qed
\end{proof}

In the following we show that the eventual leader algorithm proposed in~\cite{MRTPAA05} is self-organizing. 

\subsection{Case Study of Self-Organization: Eventual Leader}

Informally, an eventual leader service provides nodes with the identity of a node, and eventually the  identity returned by the service is the identity of a node that will remain in the system until the end of its execution. Implementing this service has received a lot of attention both in practical and theoretical area. In this work, we focus on a solution proposed by 
Raynal et al.~\cite{MRTPAA05}.  Their solution focuses on a dynamic system. 
The churn model assumed in their solution  is characterized by a  bounded dynamic demon. Specifically, 
it is assumed that there is a  set of nodes, \textsc{stable}, such that each node in \textsc{stable} after having joined the system never leaves the system nor fails. The size $\alpha$ of this set  is common knowledge, that is, each node in the system knows  that  there is a time after which $\alpha$ nodes  will never leave nor join the system. Note however that identities of these nodes are not known. 

The proposed algorithm implements the leader oracle
$\Omega$ through function \texttt{leader()}. Oracle $\Omega$  satisfies the following property :

 \begin{description}
  \item[\textbf{Eventual leadership problem~\cite{MRTPAA05} :}]
There is a time $t$ and a node  $\ell$ such that $\ell \in \textsc{stable}$, 
and after $t$, every invocation of \texttt{leader()} by any node $p$ returns $\ell$.
\end{description}

The idea of their algorithm is as follows. 
Each node $i$ maintains a set of nodes identities, denoted $\emph{trust}_i$, timestamped with a logical date locally computed ($\emph{date}_i$). Initially $\emph{trust}_i=\Pi$, with $\Pi$ the full universe of node identifiers. The aim of the algorithm is that eventually $\emph{trust}_i$ is the same for each node $i$. Specifically node $i$   repeatedly queries all the nodes of the system to determine which nodes are up and waits until it receives the first $\alpha$ responses. Then $i$ updates $\emph{trust}_i$ by making the intersection between its local $\emph{trust}_{i}$ set and the responses it has received from the other nodes. It then broadcasts  $\emph{trust}_i$ (through the invocation of a reliable broadcast primitive) to all the nodes in the system. 
A node $j$ that receives a \emph{trust} set either adopts it 
if the logical date of $\emph{trust}_{j}$ is older than the received one, or updates it by considering 
the intersection of the local and received \emph{trust} sets. If after $j$'s update, $\emph{trust}_j$
is empty then $j$ starts a new epoch with $\emph{trust}_j=\Pi$.
When \texttt{leader()} is invoked on $i$ then if   $\emph{trust}_i \neq \emptyset$ or $\emph{trust}_i \neq \Pi$ then the
smallest identifier belonging to $\emph{trust}_i$ is returned, otherwise $i$ is returned.

We now show that  the eventual leader algorithm of Raynal et al.~\cite{MRTPAA05} is self-organizing with respect to the following property:  ``Eventually \texttt{leader()}  returns the same node
to every node belonging to  the system''. 
Consider the following local criteria at node $p$:
\begin{eqnarray}\label{eq:gammaleader}
\gamma_p(t) \stackrel{\rm def}{=} (\emph{trust}_p(t),\emph{date}_p(t))
\end{eqnarray}
where $\emph{trust}_p(t)$ is the set of trust processes 
at time $t$ and $\emph{date}_p(t)$ is the logical date of the \emph{trust} set at time $t$.
\noindent
Define the partial order $\preceq$ as follows:
\begin{eqnarray}\label{eq:partialorderleader}
\gamma_p(t) \preceq \gamma_p(t^\prime) \ \mbox{if and only if} \ \emph{trust}_p(t^\prime) \subseteq \emph{trust}_p(t) 
\vee \emph{date}_p(t) \leq \emph{date}_p(t^\prime)
\end{eqnarray}

\begin{lemma}
The eventual leader algorithm \cite{MRTPAA05} is self-organizing 
with respect to  $\gamma_p$ as defined in Relation~(\ref{eq:gammaleader}) equipped with the partial order defined in Relation~(\ref{eq:partialorderleader}).
\end{lemma}

\begin{proof} 

Proof of the lemma is done by showing that the eventual leader algorithm satisfies Properties~\ref{prop:sp} and~\ref{prop:lp}.
\begin{description} 
\item[Safety]
Consider a static fragment. Three cases are possible: case \emph{i}) during that fragment, both the received  \emph{trust} set  and the local one have been updated during the same epoch. From the algorithm, the updated local  \emph{trust} set  cannot be a superset of the old one. Case \emph{ii}) during that fragment  the 
received trust set is more recent than the local one. From the algorithm, the local logical date is updated with the received one. Finally, case \emph{iii)} during that fragment, the local \emph{trust} set is empty. From the algorithm, the local date is incremented. Therefore, in all cases,  for any static fragment we have $\gamma_p(t) \preceq \gamma_p(t^\prime)$ for any $t^\prime >t$, which meets Property~\ref{prop:sp}.
\item[Liveness ] By correctness of the eventual leader algorithm~\cite{MRTPAA05} and the finite arrival churn model it 
follows that the set \emph{trust} eventually contains the same set of processes.
That is, there is a time $t$ such that $\emph{trust}_p(t^\prime)$ is stable for any $t^\prime>t$.
Since, $\emph{trust}_p(t^\prime)$ does not change, the logical date does not change either.  
Hence, $\gamma_p$ eventually reaches a $p$-stable configuration for any $p$ in the system. 
Therefore, the system is locally self-organizing which proves Property~\ref{prop:lp}. 
\end{description}
This completes the proof of the lemma. \qed
\end{proof}

\begin{corollary}
The eventual leader algorithm of Raynal et al.~\cite{MRTPAA05} is self-organizing with respect to  property:  ``Eventually \texttt{leader()}  returns the same node
to each node belonging to  the system''. 
\end{corollary}

\subsection{Strong Self-Organization}
\label{sec:strong-self-org}

The specification of both  the liveness and the weak liveness properties  do not proscribe 
nodes to reset their neighbors lists after 
C/D actions.  To prevent the system from ``collapsing''
during dynamic periods, each  node whose
neighborhood is not impacted by churn should keep local information. Such a property would  ensure that for those nodes  the global evaluation criterion restricted to these nodes  would be at least  
maintained between the end of a static fragment and the beginning of
the subsequent one. In the following this group of nodes is called the kernel
of the system.  More precisely, given two successive configurations
$c_i$ and $c_{i+1}$ with their associated graphs $G_i$ and $G_{i+1}$,
the static common core of $G_i$ and $G_{i+1}$ is the sub-graph common
to $G_i$ and $G_{i+1}$ minus all nodes for which the neighborhood has
changed. Formally, 

\begin{definition}[Topological Kernel]
Let $G1$ and $G2$ be two graphs, and
$\Gamma_{G_i}(p)$ be the set of neighbors of node $p$ in $G_i$, with $i=1,2$.  Let $\mathit{KerT}(G_1,G_2)$  be the 
topological static common core of $(G1,G2)$. Then  $$\mathit{KerT}(G_1,G_2)\stackrel{\rm def}{=}G_1\cap G_2 \setminus\{p:\Gamma_{G_1}(p)
  \neq \Gamma_{G_2}(p)\}$$
\end{definition}

Since we study systems in which self-organization may be data-oriented
(typically peer-to-peer systems), we propose a data-oriented
definition of the static core of the system. That is, given two
successive configurations $c_i$ and $c_{i+1}$, the data static common
core of $c_i,c_{i+1}$ is:

\begin{definition}[Data Kernel]
  $\mathit{KerD}(c_{i},c_{i+1})\stackrel{\rm def}{=}D_i\cap D_{i+1}$, where $D_i$ is the
  system data in $c_i$.
\end{definition}

\noindent
In the following, 
  $\mathit{Ker}^\ast$ denotes either $\mathit{KerT}$ or
  $\mathit{KerD}$. We can now state the following  property :

\begin{property}[Kernel Preservation Property]\label{prop:kernel}
  Let $\mathcal{S}$ be a system and $\gamma$ be a global evaluation
  criterion defined on the configurations of $\mathcal{S}$.  Let
  $e\!=\!(f_0, \ldots f_i,\!f_{i+1}, \ldots)$ be an execution of
  $\mathcal{S}$ and let
  $K_i=\mathit{Ker}^\ast(\fin(f_i),\debut(f_{i+1}))$.  Then, 
 \begin{eqnarray*}
\forall i,\,\,
    \gamma(\mathit{Proj}_{|K_i}(\fin(f_i))) \preceq
    \gamma(\mathit{Proj}_{|K_i}(\debut(f_{i+1})))
    \end{eqnarray*}
    where
    $\mathit{Proj}_{|K_i}(c)$ is the sub-conf\-ig\-ur\-ation of $c$
    corresponding to the kernel $K_i$.
\end{property}

\noindent
We have

\begin{definition}[Strong Self-Organization]
  \label{def:strong-self-org}
  Let $\mathcal{S}$ be a system and $\gamma$ be a global evaluation
  criterion defined on the configurations of $\mathcal{S}$.
  $\mathcal{S}$ is \emph{strongly self-organizing} for $\gamma$ if it
  is self-organizing for  $\gamma$ (cf Definition~\ref{def:self-org}) and it verifies Property~\ref{prop:kernel}.
\end{definition}

In the following we show that the One shot query algorithm presented in~\cite{baldoni07} is strongly self-organized. 

\subsection{Case Study of Strong Self-Organization: One Shot Query}

We illustrate the notion of strong-self organization through the one-shot query problem. 
The
one-time query problem, originally defined by Bawa et al.~\cite{bawa}, can
be informally described as follows.  The system is made of many nodes interconnected
through an arbitrary topology. Data is distributed among a subset of the
nodes. Some node in the system issues a query in order to
aggregate the distributed data, without knowing where the data
resides, how many processes hold the data, or even whether any data
actually matches the query. Formally, let \texttt{query(Q)} denote the operation a node invokes to aggregate
  the set of values $V=\{v_1,v_2 \ldots\}$ present in the system and
  that match the query.  Then a solution to the one shot query problem should satisfy the following two properties~\cite{bawa}: 

\begin{description}
\item[One shot query problem\cite{bawa}] Let $p$ be some node that invokes a \texttt{query(Q)} operation at time $t$.
\begin{itemize}
\item \emph{Termination Property } The \texttt{query(Q)}  operation completes in a finite time $\delta$
\item \emph{Validity Property } The set  of data $V$ returned by the \texttt{query(Q)} 
    operation includes at least all the values held by any  node $q$ such that $q$ was connected to $p$  through a sub-graph of the graph that represents the (physical) network from time $t$ to $t+\delta$.
    \end{itemize} 
\end{description}

In this paper, we consider  the solution proposed by  Baldoni et al.~\cite{baldoni07}. Their solution focuses on a dynamic system subject to a restricted kernel-based dynamic demon (cf Section~\ref{sec:dynamicschedulers}). Specifically, they assume that the physical communication 
graph $G$ guarantees a bounded  diameter not known in advance. Note that in their work, the logical graph matches $G$. In addition, they augment the system 
with a perfect failure detector~\cite{CT96}, that is a distributed oracle that allows each node to reliably update the view of its neighborhood (i.e., when node $p$  invokes its  perfect failure detector, $p$ learns which of its neighbors have  failed (or are disconnected)). 

In this model, the authors propose a DFS-based algorithm that
solves  the one-shot query specification~\cite{baldoni07}. Specifically, when  node $q$ receives a \texttt{query(Q)}  message, $q$ checks whether one of its neighbors has not already received this query  by checking both the \emph{querying} set (i.e., the set of nodes that have already sent the \texttt{query(Q)}  message and are waiting for the replies from their neighborhood), and the \emph{replied} set (i.e., the set of nodes that have provided their value upon receipt of the \texttt{query(Q)}  message). In the affirmative, $q$ sends to the first of them the \texttt{query(Q)}  message and waits until either it receives a reply from  that node, or that node is detected failed by $q$ failure detector. Node $q$ proceeds similarly with all its other neighbors. Then $q$ sends back a reply message with the set $V$ of values and the   updated \emph{replied} set  or, if node $q$ is the initiator of the query it returns only the set $V$ of values. 

 In the following we show that the DFS one-shot query algorithm of Baldoni et al.~\cite{baldoni07} is strongly self-organizing with respect to the following property: ``The initiator of a query receives at least all the values held by nodes that are connected to it through a subgraph of $G$  during the whole  query process''.

Consider the local evaluation criterion defined as follows: 
\begin{eqnarray}\label{eq:gamma_query}
\gamma_p \stackrel{\rm def}{=}\frac{\mid \emph{replied}_p \mid}{\mid \emph{target}_p \mid \setminus
\mid \emph{noResponseButFailed}_p \mid}
\end{eqnarray}
  where the set $\emph{replied}_{p}$ contains the set of nodes that answered the query,   the set $\emph{target}_p$ encompasses  the set of $p$'s neighbors (including $p$ itself) the first time $p$ received the query message, and  $\emph{noResponsebutFailed}_p$ is the set of nodes that have not yet replied but have failed since   $\emph{target}_p$ was initialized.

\begin{lemma}
  The one-query algorithm proposed by Baldoni et al~\cite{baldoni07} is strongly self-organizing
  with respect to  $\gamma_p$ as defined in Relation~(\ref{eq:gamma_query}).
\end{lemma}

\begin{proof}
  We show  that the algorithm proposed by Baldoni et al~\cite{baldoni07} verifies the kernel
  preservation property (cf. Property~\ref{prop:kernel}) and is self-organizing (cf. Definition~\ref{def:self-org}).  Let \emph{kernel} be the set of nodes which are connected to the query
  initiator during the query process through subgraph $G' \subseteq G$.  By definition this set of nodes
  does not change during instability periods.  Let $p\in \emph{kernel}$.  During instability, the evaluation criterion $\gamma_p$ evolves as follows:
  \begin{itemize}
  \item either $\gamma_p$  increases.   This is the case if and only if   \emph{i)} at least one of $p$'s neighbors replied and none of its neighbors that failed did have time to reply to $p$ or  \emph{ii)} $p$ did not received any reply and  at least one of its failed neighbors did not have time to reply before failing. 
  \item or $\gamma_p$ does not change.   
  \end{itemize}
  
 Therefore, the kernel preservation property is met. Regarding the self-organi-zing part of the proof, 
by correctness of their algorithm (cf.~\cite{baldoni}), node  $p$
  eventually receives the reply from  each of its non failed neighbors.
  Consequently,  $\gamma_p$ is eventually equal to 1, that is $\gamma_p$  reaches a
  $p$-stable configuration. By Theorem
  \ref{theo:selforg}, the algorithm is also self-organizing. \qed
\end{proof}

\begin{corollary}
The one shot query algorithm of Baldoni et al.~\cite{baldoni} is strongly self-organizing with respect to  property:  ``The initiator of a query receives at least all the values held by nodes that are connected to it through a subgraph of $G$  during the whole  proces''.
\end{corollary}

%


Finally, we can  establish a hierarchy of self-organizing classes. Basically, a class gathers a set of
self-organization properties that capture the same information about system entropy. That is,  class A is stronger than class B in the hierarchy if the entropy guaranteed by class A encompasses the entropy guaranteed by class B. 
We have:

\begin{theorem}[Self-organization Hierarchy] 
  \label{theo:hierarchy}
  weak self-organization $\subset$ self-organization $\subset$ strong
  self-organization
\end{theorem}

\begin{proof}
  Straightforward from definitions~\ref{def:weak-self-org},~\ref{def:self-org}, and ~\ref{def:strong-self-org}.  \qed
\end{proof}

\section{Composition of self-organization criteria}
\label{sec:composition}
The concept of self-organization can be easily extended to a finite
set of criteria. In the following we show that when criteria are not
interfering, i.e.,~when they are independent, then one can build a
self-organizing system for a complex criterion by using simple
criteria as building blocks.  Using the previous example where the
local evaluation criterion was proximity, a second global evaluation
criterion is needed to decrease the number of hops of a lookup
application. For instance, we may want to use a few long links to
reduce the lookup length.

\begin{definition}[Independent Criteria]
  Let $\mathcal{S}$ be a system and let $\gamma_1$ and $\gamma_2$ be
  two global criteria defined on the configurations of $\mathcal{S}$.
  Let $c$ be a configuration of $S$ and $\mathit{sc}$ and
  $\mathit{sc}^\prime$ be the sub-configurations of $c$ spanned by
  $\gamma_1$ and $\gamma_2$.  $\gamma_1$ and $\gamma_2$ are
  independent with respect to $c$ if $\mathit{sc} \neq
  \mathit{sc}^\prime$. $\gamma_1$ and $\gamma_2$ are independent with
  respect to $\mathcal{S}$ if for any configuration $c$ in
  $\mathcal{S}$, $\gamma_1$ and $\gamma_2$ are independent with
  respect to $c$.
\end{definition}

\begin{definition}[Monotonic Composition]
  \label{def:mc}
  Let $S$ be a system and let $\gamma_i \in I$ a set of criteria on
   $S$ configurations.  $\gamma=\times_{i \in I} \gamma_i$ is a
  monotonic composition of criteria $\gamma_i, i\in I$ if the
  following property is verified: $\forall c_1, c_2, \gamma(c_1) \prec
  \gamma (c_2)$ if and only if  $\exists i \gamma_i(c_1) \prec \gamma_i(c_2)$ and
  $\forall j \neq i \in I, \gamma_j(c_1) \preceq \gamma_j(c_2)$.
\end{definition}

\begin{theorem}[Multi-criteria Self-orgnization]
  \label{multi-criteria}
  Let $\mathcal{S}$ be a system and let $\gamma_1 \ldots \gamma_m$ be
  a set of independent evaluation criteria. If $S$ is weakly,
  respectively strongly, self-organizing for each $\gamma_i$, $i \in [1..m]$
  then $S$ is weakly, respectively strongly, self-organizing for $\gamma_1
  \times \ldots \times \gamma_m$.
\end{theorem}

\begin{proof}
  Let $e$ be a configuration of $S$ and let $e_i$ be the projection of
  $e$ on the sub-configurations modified by $\gamma_i$. Since $S$ is
  self-organizing with respect to $\gamma_i$ then $e_i$ is
  self-organizing with respect to $\gamma_i$.
  \begin{description}
  \item[Safety proof.]%
    Let $f$ be a static fragment in configuration~$e$ and let $f_i$ be
    the projection of $f$ on the sub-configurations spanned by
    $\gamma_i$. From the hypothesis, for all $i$,
    $\gamma_i(\debut(f_i)) \preceq \gamma_i(\fin(f_i))$, hence
    $\gamma_i(\debut(f)) \preceq \gamma_i(\fin(f))$. So,
    $\gamma(\debut(f)) \preceq \gamma(\fin(f))$.

  \item[Weak liveness proof.]%
    Let $f_i$ be a fragment. There is $f_j$ and $\gamma_k$ such that
    $\gamma_k(\debut(f_j)) \prec \gamma_k(\fin(f_j))$. Using the
    safety for all $\gamma_j, j\neq k$ it follows $\gamma(\debut(f_j))
    \prec \gamma(\fin(f_j))$.
  \end{description}

  Overall, $\mathcal{S}$ is weak self-stabilizing for $\gamma$.  The
  proof for strong self-organization follows using a similar
  reasoning.  \qed
\end{proof}

Note that Pastry is self-organizing (cf. Definition \ref{def:mc},
Theorem~\ref{multi-criteria}) for any monotonic composition of the
rounting, neighborhood and leaf set criteria.


\section{Discussions \& Open Problems}
\label{sec:conclusion}
In this paper, we have proposed a framework for defining and comparing the
self-organizing properties of dynamic systems. Self-organization is a
key feature for the newly emerging dynamic networks (peer-to-peer,
ad-hoc, robots or sensors networks). Our framework includes formal
definitions for self-organization, together with sufficient
conditions for proving the self-organization of a dynamic system. We
have illustrated our theory by proving the self-organization of two
P2P overlays: Pastry and CAN and two fundamental abstractions in
distributed computing: the leader oracle $\Omega$ and the one-shot
query problem.

We have also provided a generic algorithm that ensures the
self-organization of a system with respect to a given input criterion.
Our algorithm is based on the greedy technique, and relies solely on
the local knowledge provided by the direct neighborhood of each
process.  This algorithm can be used as a building-block in the
construction of any self-organized DHT-based or unstructured
peer-to-peer systems.

Recently, two other formal definitions for self-organization have been proposed in~\cite{BG09} and~\cite{DT09}.
In \cite{DT09} a system is self-organized if it is self-stabilizing and 
it recovers in sub-linear time following each join or leave. Furthermore, it is assumed 
that joins and leaves can cause only local changes in the system.
Obviously, this definition is too strong to capture the self-organization of P2P systems like Chord or CAN.
 In \cite{BG09} the authors extend the definition proposed in \cite{DT09} in order to capture the self-organization 
of existing P2P systems. Following \cite{DT09} a self-organizing system of $n$ processes is a system that maintains, 
improves or restores one or more safety properties following the occurrence of a subset of external actions 
involving the concurrent joins of up to $n$ processes or the concurrent departures of 
up to $\frac{n}{2}$ processes with a sub-linear recovery time per join or leave.
However, no study case is proposed.

Several problems are left open for future investigation. The first one
is the design of a probabilistic extension to our model.  This study
is motivated by the fact that connection/disconnection actions are
well-modeled by probabilistic laws. Essentially, the liveness property
could be redefined using the Markov chains model for probabilistic
dynamic I/O automata. Moreover, since our generic algorithm for
self-organization uses a greedy deterministic strategy, it may reach
just a local maximum for the global criterion. Adding randomized
choices could be a way to converge (with high probability) to a global
maximum.

  
We also intend to extend our framework towards a unified theory of the
self$^\ast$  properties of dynamic systems (i.e.,~self-healing, self-configuration, self-reconfiguration,
self-repairing).  To this end, we need
to extend our case study to other dynamic systems like robots networks
and large scale ad-hoc or sensors networks, that may offer
complementary insides for understanding the rules that govern complex
adaptive systems.


\section*{Acknowledgments}
We thank Lisa Higham and Joffroy Beauquier for their valuable comments on an earlier version
of this paper.  
\bibliographystyle{plain} \bibliography{./ref}

\end{document}